\documentclass[10pt]{article}
\usepackage[T1]{fontenc}
\usepackage{amssymb}
\usepackage{epsfig}
\usepackage{multicol}
\usepackage{graphics}
\usepackage{mathrsfs}
\usepackage{color}
\usepackage{sans}
\usepackage{amsmath}
\usepackage{amsthm}

\renewcommand{~}{\,}

\newcommand{\NV}{\Delta}

\newcommand{\e}{\mathtt{e}}

\renewcommand{\d}{\mathtt{d}}

\newcommand{\dt}{\mathtt{d}t}

\newcommand{\Att}{\mathtt{A}}
\newcommand{\B}{\mathtt{B}}
\newcommand{\al}{\alpha}
\newcommand{\be}{\beta}
\newcommand{\K}{\mathcal{K}}
\newcommand{\N}{\mathcal{N}}
\newcommand{\Z}{\mathcal{Z}}
\renewcommand{\rho}{\varrho}

\newcommand{\braw}{\langle w|}
\newcommand{\ketv}{|v\rangle}
\newcommand{\config}{\mathcal{C}}
\newcommand{\confit}{\mathcal{T}}
\newcommand{\prob}{\mathrm{P}_{\mathrm{ss}}}

\renewcommand{\geq}{\geqslant}
\renewcommand{\leq}{\leqslant}

\newcommand{\BE}{\begin{equation}}
\newcommand{\EE}{\end{equation}}

\theoremstyle{definition}
\newtheorem{lemma}{\bf Lemma}

\begin{document}

\title{Inter-particle gap distribution and spectral rigidity of totally asymmetric simple exclusion process with open boundaries}

\author{\textbf{Milan Krb\'alek}${}^{1}$ and \textbf{Pavel Hrab\'ak}${}^{1}$\\
\footnotesize $^1$ Faculty of Nuclear Sciences and Physical Engineering, Czech Technical University in Prague, Prague -- Czech Republic}

\maketitle

\begin{abstract}
We consider the one-dimensional totally asymmetric simple exclusion model (TASEP model) with open boundary conditions and present the analytical computations leading to the exact formula for distance clearance distribution, i.e. probability density for a clear distance between subsequent particles of the model. The general relation is rapidly simplified for middle part of the one-dimensional lattice using the large $N$ approximation. Both the analytical formulas and their approximations are successfully compared with the numerical representation of the TASEP model. Furthermore, we introduce the pertinent estimation for so-called spectral rigidity  of the model. The results obtained are sequentially discussed within the scope of vehicular traffic theory.\\


\end{abstract}


\section{Introduction and motivation}\label{sec:prvni}

The effect of queuing is one of commonly appearing phenomenons in nature. It can be disclosed in microbiological systems, societies of animals, computer networks, public transport systems, parking cars manoeuvres and many others. In the recent years, queuing has attracted attention of many physicists and mathematicians. It is well known that many of above-mentioned dynamics systems belong to the same class of mathematical tasks. A convenient way how to investigate such systems in detail can be found in space-discrete and time-continuous one-dimensional models based on asymmetric exclusion processes. Indeed, the family of asymmetric exclusion processes is as wide-ranging as family of its applications. Concretely, the ASEP-models have been successfully used for description of protein synthesis \cite{Schutz}, \cite{Shaw}, polymers in random media \cite{Krug}, fluctuations in shock fronts \cite{Janowsky}, \cite{Schutz2}, \cite{Derrida_Janowsky}, \cite{Mallick}, gel electronic \cite{Barkema}, in molecular biology \cite{Waterman} and finally in physics of traffic \cite{Schreckenberg}, \cite{Ha}, \cite{Headways_particle_models}, \cite{Fouladvand} or \cite{Antal}. \\

\begin{figure}[htb]
\label{fig:FD}
\begin{center}
\epsfig{file=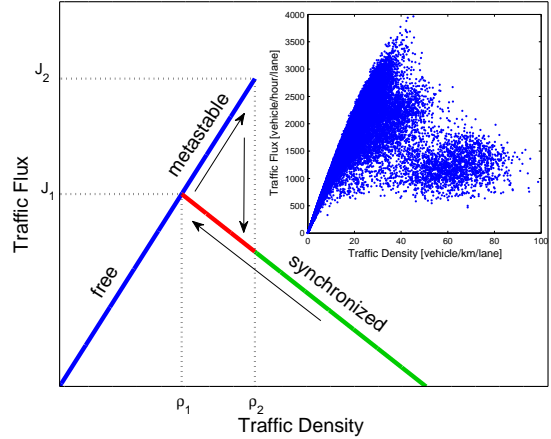,height=2.5in} \parbox{12cm}{\caption{The Empirical Flow-Density Relation and Relevant Schematic Representation.  The dependence of traffic Flux $J$ on traffic Density $\rho$ for realistic traffic flows
uses to be compared to a mirror image of the Greek letter $\lambda$ (see \cite{Review-Helbing} or \cite{Review-Chowdhury}). The relation $J=J(\rho)$ extracted from the induction-double-loop-detector data is visualized in the inset.}}
\end{center}
\end{figure}

In this article we are focused predominantly on those aspects of cellular modeling having a vehicular-traffic interpretation. With respect to the fact that macroscopical behavior of TASEP (represented for example by the fundamental relation between the flux and density of the model) corresponds to that detected in the freeway samples (compare Figures 1 and 4 for illustration) it is meaningful to consider the TASEP-model as an appropriate choice for one-lane traffic simulator. However, besides the confrontation of macroscopic quantities it is indispensable to make a comparison of microscopic quantities. Recently, the extensive investigations of traffic data (\cite{Red_cars}, \cite{Magd}, \cite{Cecile}, \cite{Treiber_PRE}, \cite{Krbalek_gas}, \cite{Wagner}, or \cite{EPJ_Helbing}) provide a good insight into the microstructure of traffic samples. It has been demonstrated in some of the above-mentioned articles that inter-vehicle gap statistics (clearance distribution) in real-road traffic can be very well estimated by the one-parametric family of functions
\BE \wp(r)=A\, \Theta(r)\,\e^{-\frac{\nu}{r}}\e^{-Br}, \label{p_beta} \EE
where
\BE B=\nu+\frac{3-\e^{-\sqrt{\nu}}}{2}, \label{becko} \EE
\BE A^{-1}=2\sqrt{\frac{\nu}{B}}~\mathcal{K}_1\bigl(2\sqrt{B\nu}\bigr).\label{acko} \EE
We remark that the functions $\Theta(x),$  $\K_\lambda(x)$ represent the Heaviside's step-function
$$\Theta(x)=\left\{\begin{array}{ccc} 1,&\hspace{0.2cm}&x>0\\0,&\hspace{0.2cm} & x\leq
0\end{array}\right.$$
and the modified Bessel's function of the second kind  (Mac-Donald's function), respectively. Furthermore, the distribution $\wp(r)$ fulfils two normalization conditions
\BE \int_\mathbb{R} \wp(r)\,\d{r}=1 \label{norma1} \EE
and
\BE \int_\mathbb{R} r\,\wp(r)\,\d{r}=1. \label{norma2} \EE
The latter represents a scaling to a mean clearance equal to one. We add that the one and only parameter $\nu$ is related to the traffic density $\rho$ (see \cite{Traffic_NV} for details). Roughly speaking, such a parameter (called speculatively as mental strain coefficient) reflects a rate of psychological pressure under which the car-drivers are during driving manoeuvres. The changes of traffic status from free flows to congested flows and vice versa are accompanied by the adequate changes of mental strain coefficient $\nu,$ i.e. by the adequate changes of clearance distribution. The chosen representatives of the relevant analysis are displayed in the Figure 2. Here one can detect the basic probabilistic trends of distances among succeeding cars. Whereas in the region of small densities (free traffic regime) the relevant probability density is exponential essentially (which fully corresponds to the fact that cars interactions are negligible) the distribution $\wp(r)$ is rapidly changing if congested data are observed. In this case the stronger mutual interactions among vehicles lead to the hardcore repulsions in the system, which results in fact that
$$\lim_{r \rightarrow 0_+} \wp(r)=0.$$
As discussed in \cite{Traffic_NV} the intermediate region of metastable traffic states shows a substantial growth of
parameter $\nu.$ This is influenced by the fact that driver, moving quite fast in relatively dense traffic flow, is under a
considerable psychological pressure. After the transition from
free to congested regime the pressure momentarily declines
because of decrease in mean velocity. Finally, if the traffic flow becomes denser and denser the mental strain coefficient $\nu$ is increasing further. This finally culminates by the creation of stop-and-go traffic waves.\\

\begin{figure}[htb]
\label{fig:TH_traffic}
\begin{center}
\epsfig{file=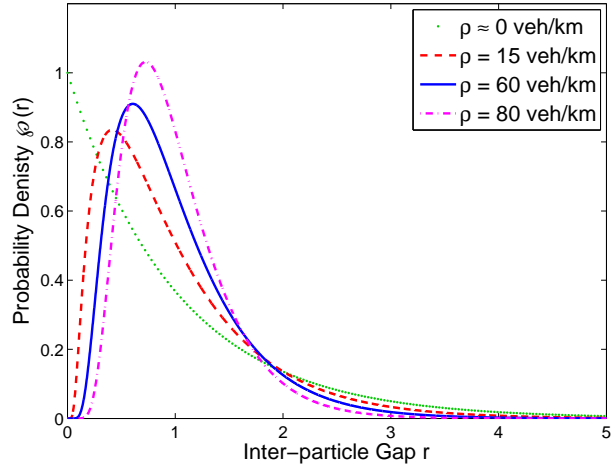,height=2.5in} \parbox{12cm}{\caption{Clearance Distribution for Real-Road Traffic. The curves visualize the analytical estimations \eqref{p_beta} of the inter-particle statistics obtained for the densities indicated in the legend.}}
\end{center}
\end{figure}

The maim goal of this article is to describe some correspondences between the microstructure of asymmetric simple exclusion model and real-road traffic. For this purpose we will compare (in the first part of this work) the relevant clearance distributions, i.e. probability densities for clear gap among all pairs of succeeding particles/vehicles. Subsequently we will analyze the associated spectral rigidities of both systems and will discuss their similarities.

\section{The totally asymmetric simple exclusion process (TASEP) with open boundaries}

Consider a chain of length $N$ containing $N$ equivalent cells and define three fixed parameters $\al,\be,p\in[0,1].$ Let each cell to be either occupied by one particle or empty. During the infinitesimal time interval $\dt$ each particle hops to the immediate site (in the defined direction) with probability $p\dt$ if the target site is empty. In the opposite case (if the target site is occupied) the particle does not change its location. If the first cell $(\ell=1)$ is empty a new particle is injected into the chain with probability $\al\dt.$ Similarly, if the last cell $(\ell=N)$ is occupied the relevant particle leaves the chain with probability $\be\dt.$ This definition can be generalized if needed, however, for all purposes the original formulation is fully sufficient. Furthermore, we will use (without loss of generality) the re-scaled variant of the model where $p=1.$\\

\begin{figure}[htb]
\begin{center}
\label{fig:def_tasep}
\epsfig{file=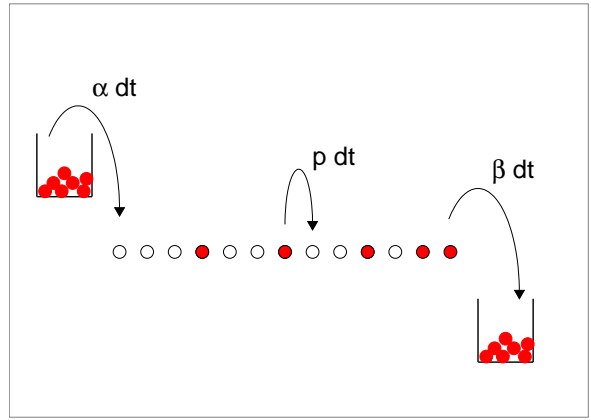,height=2.2in} \parbox{12cm}{\caption{Definition of the TASEP. The configuration visualized can be identified with the vector $\config=(0,3,1,2,1,2,1,1,2).$ }}
\end{center}
\end{figure}

The sketched rules define a simple driven lattice model of one-lane traffic whose elements are hard-core-repulsed by "implicit forces" derived from the above-mentioned exclusion rules. A great merit of such a model lies in the exact solvability of the associated steady state. Indeed, two alternative methods for exact solution of TASEP are outlined in \cite{Derrida_1}, \cite{Derrida_2}. For intentions of this research we consistently use the methods based on the \emph{matrix formulation} introduced in \cite{Derrida_2}. According to this matrix method we consider the matrices $D,E$ (infinity-dimensional, in general) and vectors $\braw$ and $\ketv$ satisfying the following algebraic rules
\BE DE=D+E,\label{DEC} \EE
\BE \braw E = \frac{1}{\al} \braw,\label{Ecko} \EE
\BE D\ketv = \frac{1}{\be}\ketv.\label{Decko} \EE
Although each configuration of the model can naturally be described by the binary sequence $$\confit=(\tau_1,\tau_2,\ldots,\tau_N) \in \{0,1\}^N,$$ it is profitable to use the following convention.
Let the symbol $\config=(n_1,m_1,n_2,m_2,\ldots,n_q,m_q)$ denote the configuration of the TASEP-chain such that (starting from the initial cell of the chain) $n_1$ is the maximal number of empty cells (i.e. the $(n_1+1)$th cell is occupied), $m_1$ is the maximal number of occupied cells (starting from the $(n_1+1)$th cell), when the $(n_1+m_1+1)$th cell is empty, and so on. Thus, $\sum_i n_i+\sum_j m_j=N.$ The previous definition (visualized for lucidity in the Figure 3) corresponds unambiguously to the relevant arrangement of the particles inside the TASEP-chain, which means that there exists some bijection between $\confit$ and $\config.$ Moreover, it has been proved in \cite{Derrida_3} that the steady-state probability of an arbitrary configuration $\config$ reads as
$$\prob(\config)= \frac{1}{\Z_N} \frac{\braw D^{n_1}E^{m_1}D^{n_2}E^{m_2}\ldots \ketv}{\langle w \ketv},$$
where the constant $\Z_N$ assures the proper normalization, i.e.
$$\sum_{\config_k} \prob(\config_k) =1$$
for $k$ running over all possible configurations. In fact, such a probability can be calculated directly using the relations eq\ref{DEC}, \eqref{Ecko}, and \eqref{Decko}, i.e.  without knowing the specific representation of matrices $D,E$ and vectors $\braw$ and $\ketv.$ Moreover, let the symbol
$$\omega_\confit(i)=\left\{\begin{array}{ccl} 1,&\hspace{0.2cm}& i-\mathrm{th~~site~~is~~occupied}\\0,&\hspace{0.2cm} & i-\mathrm{th~~site~~is~~empty},\end{array}\right.$$
represent (for the fixed configuration $\confit$) the binary functional $\omega_\confit: \widehat{N} \mapsto \{0,1\}.$ Let the symbol $$\mathscr{X}=(n_1,m_1,\ell_1,n_2,m_2,\ell_2,\ldots,n_q,m_q,\ell_q)$$ denote the set of all configurations such that $n_1$ is the maximal number of empty cells (i.e. the $(n_1+1)$th cell is occupied), $m_1$ is the maximal number of occupied cells (starting from the $(n_1+1)$th cell), when the $(n_1+m_1+1)$th cell is empty, and $\ell_1$ is the number of cells (behind the $(n_1+m_1)$th cell) which can be arbitrarily occupied or not. Precisely speaking, the symbol  $\mathscr{X}$ corresponds to the set
\begin{eqnarray*}
 && \Bigl\{(\omega_\confit(1),\omega_\confit(12),\ldots,\omega_\confit(N)) \in \{0,1\}^N:~\omega_\confit(1)=\omega_\confit(2)=\ldots=\omega_\confit(n_1)=0~\wedge~\\
 && \omega_\confit(n_1+1)=\omega_\confit(n_1+2)=\ldots=\omega_\confit(n_1+m_1)=1~\wedge~\\
 && \omega_\confit(n_1+m_1+\ell_1+1)=\omega_\confit(n_1+m_1+\ell_1+2)=\ldots=\omega_\confit(n_1+m_1+\ell_1+n_2)=0 ~\wedge~\\
 && \omega_\confit(n_1+m_1+\ell_1+n_2+1) = \omega_\confit(n_1+m_1+\ell_1+n_2+2)=\ldots=\omega_\confit(n_1+m_1+\ell_1+n_2+m_2)=1 ~\wedge~\\
 && \omega_\confit(n_1+m_1+\ell_1+n_2+m_2+\ell_2+1)=\ldots=\omega_\confit(n_1+m_1+\ell_1+n_2+m_2+\ell_2+n_3)=0 ~\wedge~\ldots\Bigr\}.
\end{eqnarray*}
As follows from the results published for steady state of TASEP-model (see for example \cite{Derrida_3} or \cite{Blythe}) the probability for finding the system in one of configurations included in $\mathscr{X}$ is
\BE
\prob(\mathscr{X})= \frac{1}{\Z_N} \frac{\braw D^{n_1}E^{m_1}(D+E)^{\ell_1} D^{n_2}E^{m_2}(D+E)^{\ell_2}\ldots \ketv}{\langle w \ketv}.\label{velka_formule}
\EE
As a direct consequence of this assertion we easily deduce that if $\mathscr{X}=(0,0,N)$ then
$$\prob(\mathscr{X})= \frac{1}{\Z_N} \frac{\braw D^0E^0(D+E)^N \ketv}{\langle w \ketv}= \frac{1}{\Z_N} \frac{\braw (D+E)^N \ketv}{\langle w \ketv}.$$
As $\mathscr{X}=(0,0,N)$ includes all permissible configurations of the system (and therefore $\prob(\mathscr{X})=1$) one can trivially calculate the value of the respective partition sum as
$$\Z_N= \frac{\braw (D+E)^N \ketv}{\langle w \ketv}.$$
Using the lemma \ref{lemma:den} and notation~\eqref{eq:notation} we can assert that
$$
\Z_N= \frac{1}{\langle w \ketv} \sum_{m=1}^N \frac{m(2N-m-1)!}{N!(N-m)!}\sum_{i=0}^m \braw E^i~D^{m-i} \ketv =\sum_{m=1}^N \B_{N,m}\sum_{i=0}^m \frac{1}{\al^i}\frac{1}{\be^{m-i}}.
$$
Hence
\BE \Z_N=\left\{\begin{array}{lll} \sum_{m=1}^N \B_{N,m} \frac{\be^{-m-1}-\al^{-m-1}}{\be^{-1}-\al^{-1}} & \quad & \al\neq \be \\
 \sum_{m=1}^N \B_{N,m} m\al^{-m} & \quad & \al=\be. \end{array} \right. \label{normalka} \EE

We conclude that by means of the formulas \eqref{velka_formule} and \eqref{normalka} one can enumerate the probability of an arbitrary steady-state-configuration of TASEP.

\section{One-dimensional representation of associated matrix algebra}

As demonstrated above, the steady-state probability distribution of an arbitrary arrangement of TASEP particles is derived from the matrix algebra \eqref{DEC} -- \eqref{Decko}. Such a quadratic algebra is formulated for two matrices $D$ and $E$ which are associated to the particles and holes respectively. As proven in \cite{Derrida_2} the matrices fulfilling the rules \eqref{DEC} -- \eqref{Decko} are non-commuting and infinity-dimensional, in general. However, for special choice of parameter $\al$ and $\be$ the matrices $E$ and $D$ can be commuting, i.e. $DE=ED.$ Under this condition the equation \eqref{DEC} leads to the equations
\BE \label{eq:D+E=DE} D+E=DE=ED, \EE
$$\frac{1}{\be}\ketv + E\ketv = E \frac{1}{\be}\ketv$$
$$\frac{1}{\be}\langle w \ketv + \frac{1}{\al}\langle w\ketv = \braw \frac{1}{\al\be}\ketv$$
\BE \al+\be=1. \label{special_condition} \EE
Thus, if the condition \eqref{special_condition} is guaranteed the commuting matrices $D$ and $E$ are one-dimensional, i.e. they are represented in fact by the numbers $E=\al^{-1}$ and $D=\be^{-1}.$ Consecutively, we can choose $\braw=\ketv=1.$ Above that, as the partition sum is now reduced to
$$\Z_N= \frac{\braw (D+E)^N \ketv}{\langle w \ketv}= \left(\frac{1}{\al\be}\right)^N$$
the steady-state probability (for the set of configurations $\mathscr{X}$) reads as
$$\prob(\mathscr{X})= (\al\be)^N  \left(\frac{1}{\be}\right)^{\Sigma n_i}\left(\frac{1}{\al}\right)^{\Sigma m_i}\left(\frac{1}{\al\be}\right)^{\Sigma \ell_i}=\al^{N-\Sigma m_i-\Sigma \ell_i}\be^{N-\Sigma n_i-\Sigma \ell_i}=\al^{\Sigma n_i}(1-\al)^{\Sigma m_i}.$$

\section{Macroscopic characteristics in mean-field approximation}

For intentions of this article we briefly summarize the known results on macroscopical behavior of TASEP-model. We are concentrated predominantly on the bulk density and flux in mean-field approximation. Concretely, the average density of particles occurring inside the $i$th cell can be depicted as a probability for finding the system in the certain configuration chosen from $\mathscr{X}_{\rho_i}=(0,0,i-1,1,0,N-i).$ That means
$$\rho_i^{(N)}=\prob(\mathscr{X}_{\rho_i})=\frac{\braw (D+E)^{i-1}D(D+E)^{N-i}\ketv}{\braw (D+E)^N\ketv}\,.$$
Similarly, the flux through the $i$th site is defined as
$$J_i^{(N)}=\prob(\mathscr{X}_{J_i})=\frac{\braw (D+E)^{i-1}DE(D+E)^{N-i-1}\ketv}{\braw (D+E)^N\ketv}=\frac{\Z_{N-1}}{\Z_N}\,,$$
where $\mathscr{X}_{J_i}=(0,0,i-1,1,1,N-i-1)$. Using the large $N$ approximations 
\BE
\Z_N \approx \left\{ \begin{array}{lll}
\frac{\be\al(1-2\al)}{\be-\al}\left(\frac{1}{\al(1-\al)}\right)^{N+1} & \hspace{0.2cm} & \al < \frac{1}{2}~\wedge~\al<\be\\
\frac{\be\al(1-2\be)}{\al-\be}\left(\frac{1}{\be(1-\be)}\right)^{N+1} & \hspace{0.2cm} & \be < \frac{1}{2}~\wedge~\be<\al\\
\frac{\be\al}{\sqrt{\pi}(\al-\be)}\left[\frac{1}{(2\al-1)^2}-\frac{1}{(2\be-1)^2}\right]\frac{4^N}{N^{3/2}} & \hspace{0.2cm} & \al>\frac{1}{2}~\wedge~\be>\frac{1}{2}\\
\frac{\al^2}{\sqrt{\pi}(2\al-1)^3}\frac{4^{N+1}}{N^{3/2}} & \hspace{0.2cm} & \al=\be>\frac{1}{2}\\
\frac{(1-2\al)^2}{(1-\al)^2}\frac{N}{\al^N(1-\al)^N} & \hspace{0.2cm} & \al=\be<\frac{1}{2}\\
\frac{2\be}{\sqrt{\pi}(2\be-1)}\frac{4^{N}}{N^{1/2}} & \hspace{0.2cm} & \al=\frac{1}{2}<\be\\
4^N & \hspace{0.2cm} & \al=\be=\frac{1}{2}\\
\end{array} \right.\label{Denzitka}
\EE
derived in \cite{Derrida_2} we ascertain that (for an arbitrary cell being far from the boundaries of the system) the relevant bulk-density is given by
\BE
\label{eq:rholim}
\rho(\al,\be) = \left\{ \begin{array}{lll}
\frac{1}{2} & \hspace{0.2cm} & \al\geq \frac{1}{2}~\wedge~\be\geq \frac{1}{2},\\
\al & \hspace{0.2cm} & \al < \frac{1}{2}~\wedge~\be>\al,\\
1-\be & \hspace{0.2cm} & \be< \frac{1}{2}~\wedge~\be<\al,\\
\al & \hspace{0.2cm} & \al+\be=1.
\end{array} \right.
\EE
Regarding the TASEP flux $J$ it is trivial to show that it is independent on $i$ and
\BE
J(\al,\be) = \left\{ \begin{array}{lll}
\frac{1}{4} & \hspace{0.2cm} & \al\geq \frac{1}{2}~\wedge~\be\geq \frac{1}{2},\\
\al(1-\al) & \hspace{0.2cm} & \al < \frac{1}{2}~\wedge~\be>\al,\\
\be(1-\be) & \hspace{0.2cm} & \be< \frac{1}{2}~\wedge~\be<\al,\\
\al(1-\al) & \hspace{0.2cm} &  \al+\be=1.
\end{array} \right.\label{Fluxka}
\EE
After comparing the results \eqref{Denzitka} and \eqref{Fluxka} we close this section by the assertion that the fundamental diagram of TASEP-model (see Figure 4) is described  by the equation
\BE J(\rho)=\rho(1-\rho).\label{Fundam_TASEP} \EE

\begin{figure}[htb]
\begin{center}
\epsfig{file=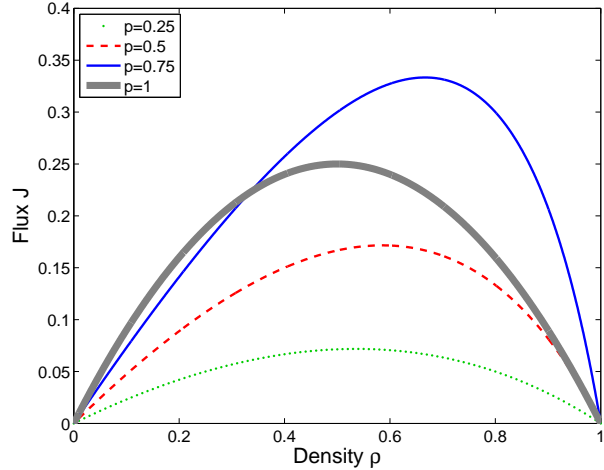,height=2.5in} \parbox{12cm}{\caption{The Flow-Density Relation for TASEP-model. The curves display the function $J=J(\rho)$ calculated for ordered-sequential update of ASEP-model with periodic boundary conditions (for details on TASEP update-procedures, please, see \cite{Rajewsky}). Parameter $p$ represents the rate for particle transition from the actual cell to the following one. The bold curve displays the theoretical relation \eqref{Fundam_TASEP} valid for time-continuous TASEP. \label{fig:tasepfundam}}}
\end{center}
\end{figure}

\section{Distance clearance distribution - a special case}

In the main part of this article we are aiming to derive an exact analytical formula for the so-called distance clearance of TASEP-model, i.e. probability density for clear distance (measured in cells) among the subsequent particles of the model. Such a probability density can be (in general) depending on the position $i$ of a monitored particle. Hence, the relevant distribution $\wp_i(k)$ is depending on $i$ as well. The unnormalized probability of a $k$cell-gap behind the $i$th occupied cell is denoted as $\widetilde{\wp_i}(k)=\prob(0,0,i-1,1,k-1,0,1,0,N-i-k)$ and can be calculated via
\BE
    \widetilde{\wp_i}^{(N)}(k)=\frac{1}{\Z_N} \frac{\braw (D+E)^{i-1}DE^{k-1}D(D+E)^{N-i-k} \ketv}{\langle w \ketv}\,\qquad (1\leq k\leq N-i),\label{hlavni_vztah}
\EE
because the first $i-1$ cells could be randomly occupied or vacant, the $i$th cell should be occupied, $k-1$ succeeding cells should be vacant as well, the $(i+k)$th cell should be occupied, and finally the tail of the chain could be occupied/vacant at random. The probability density for clearance is then $\wp_i^{(N)}(k)=\N\widetilde{\wp_i}^{(N)}(k)$, where the constant $\N$ assures the proper normalization computable via the condition
\BE
\sum_{k=1}^{N-i} \wp_i^{(N)}(k)=1.\label{headway_normalization}
\EE
If the associated matric algebra is one-dimensional, i.e. under the condition \eqref{special_condition}, then \eqref{hlavni_vztah} changes to
\BE
\wp_i^{(N)}(k)=\N \left(\frac{1}{\al}\right)^{N-2}\left(\frac{1}{\be}\right)^{N+1-k} \quad (1\leq k\leq N-i).\label{headway_special}
\EE
Precisely speaking, the derivation of the last formula using the relation \eqref{hlavni_vztah} is legitimate for $k \geq 3,$ since for $k=1,2$ the relation \eqref{hlavni_vztah} has a slightly modified form. However, it can easily be verified that relation \eqref{headway_special} holds true for all positive $k.$ The normalization pre-factor $\N$ can be quantified either from condition \eqref{headway_normalization} or from the equation
$$\N^{-1} = \frac{\braw (D+E)^{i-1}D(D+E)^{N-i} \ketv - \braw (D+E)^{i-1}DE^{N-i} \ketv}{\langle w \ketv}$$
that takes into account all possible configurations with the occupied $i$th cell and rejects those configurations having the empty tail (behind the $i$th cell), when all the cells $i+1,~i+2,\ldots,N$ are vacant. This leads to
$$\N^{-1} = \left(\frac{1}{\al}\right)^{N-1}\left(\frac{1}{\be}\right)^{N} \left(1-\be^{N-i}\right),$$
which yields (if \eqref{headway_special} is applied) the result desired
$$\wp_i^{(N)}(k)=\frac{\al}{1-\be^{N-i}}~\be^{k-1}.$$
The large $N$ approximation adjusts such a expression to
$$\wp_i(k)=\lim_{N\rightarrow+\infty}\wp_i^{(N)}(k)=\al~\be^{k-1}.$$
Note that this result is independent of the cell behind which the clearance is measured. As a consequence, we obtain the formula
$$ \langle k \rangle =\sum_{k=1}^\infty k \wp_i(k) = \al \sum_{k=1}^\infty k\be^{k-1}= \frac{\al}{(1-\be)^2}=\frac{1}{\al}$$
for mean clearance of TASEP-model. If reformulated the clearance distribution reads as
$$\wp(k)=\rho(1-\rho)^{k-1},$$
where the relation $\rho=\al$ is adopted. We remind that these outcomes are valid under the condition $\al+\be=1$ only.

\section{Distance clearance distribution - a general case}

To derive the exact formula of $\wp_i^{(N)}(k)$ for an arbitrary choice of parameters $\al$, $\be$ we will first rearrange the matrix product in~\eqref{hlavni_vztah} into the form
$$ (D+E)^{i-1}DE^{k-1}D(D+E)^{N-i-k}=\sum_{m,n}a_{m,n}E^mD^n.$$
Then it will be hold
$$\braw(D+E)^{i-1}DE^{k-1}D(D+E)^{N-i-k}\ketv=\sum_{m,n}a_{m,n}\al^{-m}\be^{-n}.$$
Using the notation~\eqref{eq:notation}, lemma~\ref{lemma:d+en}, and lemma~\ref{lemma:den} we get (for $k\neq N-i$)
\begin{align*}
    D(D+E)^{N-i-k}=\sum_{m=1}^{N-i-k}\B_{N-i-k,m}\sum_{\ell=0}^{m}DE^\ell D^{m-\ell}=\sum_{m=1}^{N-i-k}\B_{N-i-k,m}\sum_{\ell=0}^{m}\left(D^{m-\ell+1}+\sum_{z=1}^\ell E^zD^{m-\ell}\right).
\end{align*}
Furthermore, by means of lemma~\ref{lemma:den} we obtain
\begin{multline*}
    DE^{k-1}D(D+E)^{N-i-k}=\sum_{m=1}^{N-i-k}\B_{N-i-k,m}\sum_{l=0}^{m}\left(DE^{k-1}D^{m-\ell+1}+\sum_{z=1}^{\ell}DE^{z+k-1}D^{m-\ell}\right)=\\
    =\sum_{m=1}^{N-i-k}\B_{N-i-k,m}\sum_{\ell=0}^{m}\left(D^{m-\ell+2}+\sum_{w=1}^{k-1}E^wD^{m-\ell+1}+\ell D^{m-\ell+1}+\sum_{z=1}^{\ell}\sum_{w=1}^{z+k-1}E^wD^{m-\ell}\right).
\end{multline*}
For $i\geq2$ we apply the lemma~\ref{lemma:d+en} in the product $(D+E)^{i-1}.$ That provides
\begin{multline*}
    (D+E)^{i-1}DE^{k-1}D(D+E)^{N-i-k}=\sum_{p=1}^{i-1}\sum_{q=0}^{p}\sum_{m=1}^{N-i-k}\sum_{\ell=0}^{m}\B_{N-i-k,m}\B_{i-1,p}
    \Bigl(E^qD^{p-q+m-\ell+2}+\\ +\sum_{w=1}^{k-1}E^q\underbrace{D^{p-q}E^w}_{}D^{m-\ell+1}+\ell E^qD^{p-q+m-\ell+1}+\sum_{z=1}^{\ell}\sum_{w=1}^{z+k-1}E^q\underbrace{D^{p-q}E^w}_{}D^{m-\ell}\Bigr).
\end{multline*}
From the lemma~\ref{lemma:dnem} we derive that (for $q<p$)
$$ D^{p-q}E^w=\sum_{a=1}^{w}\frac{(p-q-a+w-1)!}{(p-q-1)!(w-a)!}E^a+\sum_{b=1}^{p-q}\frac{(p-q-b+w-1)!}{(p-q-b)!(w-1)!}D^b
    =\sum_{a=1}^{w}\Att_1E^a+\sum_{b=1}^{p-q}\Att_2D^b,$$
where we use, for convenience, the notation
$$ \Att_1=\frac{(p-q-a+w-1)!}{(p-q-1)!(w-a)!}\,,\qquad \Att_2=\frac{(p-q-b+w-1)!}{(p-q-b)!(w-1)!}.$$
This leads to the final formula
\begin{multline*}
    (D+E)^{i-1}DE^{k-1}D(D+E)^{N-i-k}=\\
    =\sum_{p=1}^{i-1}\sum_{m=1}^{N-i-k}\sum_{\ell=0}^{m}\sum_{q=0}^{p-1}\B_{N-i-k,m}\B_{i-1,p}\left(E^qD^{p-q+m-\ell+2}+\ell E^qD^{p-q+m-\ell+1}+\sum_{w=1}^{k-1}\sum_{a=1}^{w}\Att_1E^{q+a}D^{m-\ell+1}+\right.\\
    \left.+\sum_{w=1}^{k-1}\sum_{a=1}^{p-q}\Att_2E^qD^{m-\ell+1+b}+\sum_{z=1}^{\ell}\sum_{w=1}^{z+k-1}\sum_{a=1}^{w}\Att_1E^{q+a}D^{m-\ell}+\sum_{z=1}^{\ell}\sum_{w=1}^{z+k-1}\sum_{b=1}^{p-q} \Att_{2} E^{q} D^{m-\ell+b}\right)+\\
    +\sum_{p=1}^{i-1}\sum_{m=1}^{N-i-k}\sum_{\ell=0}^{m}\B_{N-i-k,m}\B_{i-1,p}\left(E^pD^{m-\ell+2}+\sum_{w=1}^{k-1}E^{w+p}D^{m-\ell+1}+\ell E^pD^{m-\ell+1}+\sum_{z=1}^{\ell}\sum_{w=1}^{z+k-1}E^{w+p}D^{m-\ell}\right).
\end{multline*}
Multiplying such an expression by vectors $\braw$ and $\ketv$ gives the following dependence on parameters $\al,$ $\be$
\begin{multline}
 \wp_i^{(N)}(k)
 =\frac{\mathcal{A}}{Z_N\langle w\ketv}\sum_{p=1}^{i-1}\sum_{m=1}^{N-i-k}\sum_{\ell=0}^{m}\B_{N-i-k,m}\B_{i-1,p}\left(\frac{1}{\beta}\right)^{m-\ell}\left\{\left(\frac{1}{\beta^2}+ \frac{\ell}{\beta}\right)\sum_{q=0}^{p-1}\left(\frac{1}{\alpha}\right)^q\left(\frac{1}{\beta}\right)^{p-q} +\right.\\
 +\sum_{q=0}^{p-1}\left(\frac{1}{\alpha}\right)^q\sum_{w=1}^{k-1}\frac{1}{\beta}\left(\sum_{a=1}^{w}\Att_1\left(\frac{1}{\alpha}\right)^{a} +\sum_{b=1}^{p-q}\Att_2\left(\frac{1}{\beta}\right)^{b}\right)+\\
 +\sum_{z=1}^{\ell}\sum_{q=0}^{p-1}\left(\frac{1}{\alpha}\right)^{q}\sum_{w=1}^{z+k-1}\left(\sum_{a=1}^{w}\Att_1\left(\frac{1}{\alpha}\right)^{a} +\sum_{b=1}^{p-q} \Att_{2}  \left(\frac{1}{\beta}\right)^{b}\right)+\\
 +\left(\frac{1}{\alpha}\right)^p\left.\left(\frac{1}{\beta^2}+\frac{\ell}{\beta} +\sum_{w=1}^{k-1}\left(\frac{1}{\alpha}\right)^{w}\frac{1}{\beta}+ \sum_{z=1}^{\ell}\sum_{w=1}^{z+k-1}\left(\frac{1}{\alpha}\right)^{w} \right) \right\}\,, \label{finalka-clearance}
\end{multline}
where we consider $i\geq2$ and $k<N-i$. The formulas for the remaining cases have been obtained analogically.
$$\wp_1^{(N)}(k)=\frac{\mathcal{A}}{Z_N\langle w\ketv}\sum_{m=1}^{N-i-k}\sum_{\ell=0}^{m}\B_{N-i-k,m}\left(\frac{1}{\beta}\right)^{m-\ell}\left\{\frac{1}{\beta^2}+\frac{\ell}{\beta} +\sum_{w=1}^{k-1}\left(\frac{1}{\alpha}\right)^{w}\frac{1}{\beta}+ \sum_{z=1}^{\ell}\sum_{w=1}^{z+k-1}\left(\frac{1}{\alpha}\right)^{w} \right\},$$
\begin{multline*}
    \wp_i^{(N)}(N-i)=\frac{\mathcal{A}}{Z_N\langle w\ketv}\sum_{p=1}^{i-1}\B_{i-1,p}\left\{\sum_{q=0}^{p}\left(\frac{1}{\al}\right)^q\left(\frac{1}{\beta}\right)^{p-q+1}+\sum_{w=1}^{N-i-1}\left(\frac{1}{\al}\right)^{p+w}+\right.\\
    \left.+\sum_{q=0}^{p-1}\left(\frac{1}{\al}\right)^q\sum_{w=1}^{N-i-1}\left(\sum_{a=1}^{w}\Att_1\left(\frac{1}{\alpha}\right)^{a} +\sum_{b=1}^{p-q}\Att_2\left(\frac{1}{\beta}\right)^{b}\right)\right\},
\end{multline*}
$$ \wp_1^{(N)}(N-1)=\left(\frac{1}{\be}\right)^2+\sum_{w=1}^{N-2}\left(\frac{1}{\al}\right)^w\left(\frac{1}{\be}\right).$$

Now we are focused on the asymptotic clearance distribution $\wp(k)$. As we are interested in the probability of a gap of the length $k$ between two succeeding particles in the bulk, i.e. far from the boundaries, it is beneficial to fix the position of cells $i,i+1,\dots,i+k$. Based on this strategy we denote the number of arbitrary occupied cells before the gap as $m=i-1$ and the number of arbitrary occupied cells behind the gap by $n=N-i-k$. The limit $N\rightarrow\infty$ is then meant in the sense $m,n\rightarrow\infty$. This approach assures that the gap is located far from the boundaries with respect to the lattice size $N,$ which is necessary for the calculations below.\\

Let us now investigate the behavior of the quantity $\widetilde{\wp}_{m,n}(k)=\widetilde{\wp}_{m+1}^{(N)}(k)$, where $N=m+n+k+1$. In the term of matrix-product-ansatz description (MPA) we focuss on
\BE
\label{eq:wpmn} \widetilde{\wp}_{m,n}(k)=\frac{1}{\Z_{m+n+k+1}} \frac{\braw (D+E)^{m}DE^{k-1}D(D+E)^{n} \ketv}{\langle w \ketv}
\EE
in the limit $m,n\rightarrow\infty$. Firstly we derive the asymptotic form $\widetilde{\wp}(k)=\lim_{m,n\rightarrow\infty}\widetilde{\wp}_{m,n}(k)$ for $k=1,2$ and then we get a recursive formula for $\widetilde{\wp}(k+1)$, $k\geq3$. Rearranging the matrix product in~\eqref{eq:wpmn} using~\eqref{eq:D+E=DE} for $k=1$ we get
\begin{multline*}
    (D+E)^{m}DD(D+E)^{n}=(D+E)^{m}D(D+E)^{n+1}-(D+E)^{m}DE(D+E)^{n}=\\
    =(D+E)^{m}D(D+E)^{n+1}-(D+E)^{m+n+1}.
\end{multline*}
And hence, using the asymptotic form of fundamental dependence $J(\rho)=\rho(1-\rho)$, it holds
\BE
\label{eq:pk1} \widetilde{\wp}_{m,n}(1)=\rho_{m+1}^{(n+m+2)}-J^{(n+m+2)}\quad\Rightarrow\quad\widetilde{\wp}(1)=\rho-\rho(1-\rho)=\rho^2.
\EE
For $k=2$ we analogically obtain
$$ (D+E)^{m}DED(D+E)^{n}=(D+E)^{m+1}D(D+E)^{n},$$
and therefore
\BE
\label{eq:pk2} \widetilde{\wp}_{m,n}(2)=\rho_{m+1}^{(m+n+2)}J^{(m+n+3)}\quad\Rightarrow\quad\widetilde{\wp}(2)=\rho\rho(1-\rho)=\rho^2(1-\rho).
\EE
Moreover, for $k\geq3$ it holds
$$ (D+E)^{m}DE^{k-1}D(D+E)^{n}=(D+E)^{m}DE^{k-2}D(D+E)^{n}+(D+E)^{m-1}DE^{k}D(D+E)^{n}$$
That leads to the formula
\BE
\label{eq:recursion} \widetilde{\wp}_{m-1,n}(k+1)=\widetilde{\wp}_{m,n}(k)-J^{m+n+k+1}\widetilde{\wp}_{m,n}(k-1).
\EE
For an arbitrary fixed $k\in\mathbb{N}$ it holds that
$$ \lim_{m,n\rightarrow\infty}\widetilde{\wp}_{m,n}(k)=\lim_{m,n\rightarrow\infty}\widetilde{\wp}_{m-1,n}(k)=:\widetilde{\wp}(k).$$
The existence of the limit $\widetilde{\wp}(k)$ for all $k\in\mathbb{N}$ can be elementarily confirmed using mathematical induction by means of relations~\eqref{eq:pk1},~\eqref{eq:pk2},~\eqref{eq:recursion}. Now we can apply the limit $m,n\rightarrow\infty$ on the relation~\eqref{eq:recursion} obtaining
\begin{align*}
    &\widetilde{\wp}(k+1)=\widetilde{\wp}(k)-\rho(1-\rho)\widetilde{\wp}(k-1)\,,&&\widetilde{\wp}(1)=\rho^2\,,&&\widetilde{\wp}(2)=\rho^2(1-\rho).&
\end{align*}
It is easy to verify that the solution of this recursion is the relation
$$\widetilde{\wp}(k)=\rho^2(1-\rho)^{k-1}.$$
The normalization constant $\N$ is then
$$\N=\left(\sum_{k=1}^{\infty}\rho^2(1-\rho)^{k-1}\right)^{-1}=\frac{1}{\rho}$$
and hence the clearance distribution $\wp(k)=\N\widetilde{\wp}(k)$ fulfils the relationship
\BE
    \wp(k)=\varrho(1-\varrho)^{k-1}. \label{universal_headway_distribution}
\EE
All the time we kept in mind that the bulk density $\rho$ is a function of parameters $\al$, $\be$ given by~\eqref{eq:rholim}. Thus, the dependence of the clearance distribution on these parameters reads
$$
    \wp(k;\al,\be)=
    \begin{cases}
        \frac{1}{2^k}&  \al\geq \frac{1}{2}~\wedge~\be\geq \frac{1}{2},\\
        \al(1-\al)^{k-1}&   \al < \frac{1}{2}~\wedge~\be>\al,\\
        (1-\be)\be^{k-1}&   \be< \frac{1}{2}~\wedge~\be<\al,\\
        \al\be^{k-1}&   \al+\be=1.
    \end{cases}
$$

\section{Mean-field spectral rigidity of TASEP-model}

Description of traffic microstructure by means of inter-vehicle gap distribution $\wp(r)$ is, as discussed above, usual in traffic theory. However, clearance distribution depicts the distance gaps between two successive vehicles only. Aiming to investigate the middle-ranged interactions among the cars it is necessary to find a mathematical quantity suitable for quantifying the level of synchronization for larger clusters of particles. This desired quantity can be found in Random Matrix Theory (see \cite{Mehta}) where it provides an insight into the structure of eigenvalues of random matrix ensembles. It is called \emph{a spectral rigidity.} If reformulated within the bounds of traffic theory the spectral rigidity has the following interpretation.\\

\begin{figure}[htb]
\begin{center}
\epsfig{file=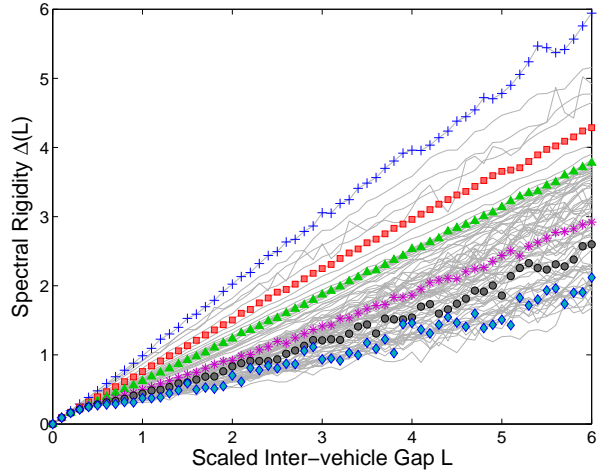,height=2.5in} \parbox{12cm}{\caption{Spectral Rigidity of Vehicular Traffic. Gray curves correspond to the rigidity $\Delta(L)$ analyzed separately in 85 density regions. The chosen results of relevant statistical analysis are picked out. Concretely, plus signs, squares, triangles, stars, circles, and diamonds represent the rigidities obtained for traffic data from the following density regions: [2,3); [6,7); [16,17); [26,27); [40,41); and [85,86) vehicles/km/lane respectively. For details, please, see the article \cite{Traffic_NV}.}}
\end{center}
\label{fig:vehicular_rigidity}
\end{figure}

Consider a set $\{r_i:i=1 \ldots Q\}$ of netto gaps between each pair of the succeeding cars moving
on an one-lane-freeway. We suppose that the mean gap taken over the
complete set is re-scaled to one, i.e.
\BE
\sum_{i=1}^Q r_i=Q.\label{preskalovani}
\EE
Dividing the interval $[0,Q)$ into subintervals
$\bigl[(k-1)L,kL\bigr)$ of a length $L$ and denoting by $n_k(L)$
the number of cars in the $k-$th subinterval, the average value
$\overline{n}(L)$ taken over all possible subintervals is
$$\overline{n}(L)=\frac{1}{\lfloor Q/L \rfloor} \sum_{k=1}^{\lfloor Q/L
\rfloor} n_k(L)=L,$$
where the integer part $\lfloor Q/L \rfloor$ stands for the number
of all subintervals $\bigl[(k-1)L,kL\bigr)$ included in the entire interval
$[0,Q).$ We suppose, for convenience, that $Q/L$ is integer, i.e.
$\lfloor Q/L \rfloor=Q/L.$ The spectral rigidity $\Delta(L)$ is then
defined as
$$\NV(L)=\frac{L}{Q} \sum_{k=1}^{Q/L} \bigl(n_k(L)-L\bigr)^2$$
and represents the statistical variance of the number of vehicles
moving at the same time inside a fixed part of the road of a length
$L.$
For the one-parametric family of distributions \eqref{p_beta} there has been proved in the article \cite{Traffic_NV} that the associated spectral rigidity has a form
\BE
\NV(L)=\chi L +\gamma + \mathcal{O}(L^{-1}),\label{forestGUMP}
\EE
where (together with the relation \eqref{becko})
\BE
\chi=\chi(\nu)=\frac{2+\sqrt{B\nu}}{2B(1+\sqrt{B\nu})}\label{forestGUMP_chi}
\EE
and
\BE
\gamma=\gamma(\nu)=\frac{6\sqrt{B\nu}+B\nu\bigl(21+4B\nu+16\sqrt{B\nu}\bigr)}{24\bigl(1+\sqrt{B\nu}\bigr)^4}.\label{forestGUMP_gamma}
\EE
It means that the rigidity is a linear function whose slope $\chi$ is depending on the mental strain coefficient $\nu$ (briefly described in the section \ref{sec:prvni}). Approximately, the changes of $\chi$ (with respect to the traffic density) can be described as follows. For small densities the slope $\chi$ in relation $\NV(L) \approx \chi L +\gamma$ is close to one as expected for statistically independent events. Nevertheless, if the traffic density increases the interactions among vehicles strengthen, which results in a descent of the slope $\chi.$ The more detailed insight into the realistic behavior of $\NV(L)$ is demonstrated in the Figure 5, where the experimental results of "traffic spectral analysis" are plotted.\\

At the moment our immediate goal is to derive the analytical formula for the spectral rigidity $\NV$ of totally asymmetric simple exclusion model. Firstly, note that the mean clearance calculated for the distribution \eqref{universal_headway_distribution} is
$$\langle k \rangle =\sum_{k=1}^\infty k \rho(1-\rho)^{k-1}= \frac{1}{\rho},$$
which means that in the end of our computations there will be necessity to revise the results in accord with the general definition \eqref{preskalovani}. Now, denote by $n_\ell(s)$ the probability that there is exactly $s$ particles inside the fixed $\ell-$cells region of TASEP-chain. As understandable, such a probability reads
$$n_\ell(s)= {\ell \choose s} \rho^s (1-\rho)^{\ell-s}, \quad (s=0,1,\ldots,\ell).$$
Since
$$\langle s \rangle = \sum_{s=0}^\ell \frac{\ell!}{(\ell-s)!(s-1)!} \rho^s (1-\rho)^{\ell-s}= \rho \ell$$
and
$$\langle s^2 \rangle = \sum_{s=0}^\ell \frac{s~\ell!}{(\ell-s)!(s-1)!} \rho^s (1-\rho)^{\ell-s}= \rho \ell \bigl(1+\rho(\ell-1)\bigr),$$
the corresponding statistical variance is of a form
$$\sigma^2= \langle s^2 \rangle - \langle s \rangle^2= \rho \ell(1-\rho),$$
which leads (after the transition to the re-scaled headways $L=\rho\ell$) to the final formula for the spectral rigidity
\BE \NV(L)=(1-\rho)L. \label{NV-final-formula}\EE
Obviously, the slope $\chi=1-\rho$ of the rigidity $\NV(L)$ decreases linearly with increasing density of the TASEP-particles, as confirmed by the numerous numerical tests visualized in the Figure 6. This result is very similar to the behavior revealed in realistic traffic samples.

\begin{figure}[htb]
\begin{center}
\epsfig{file=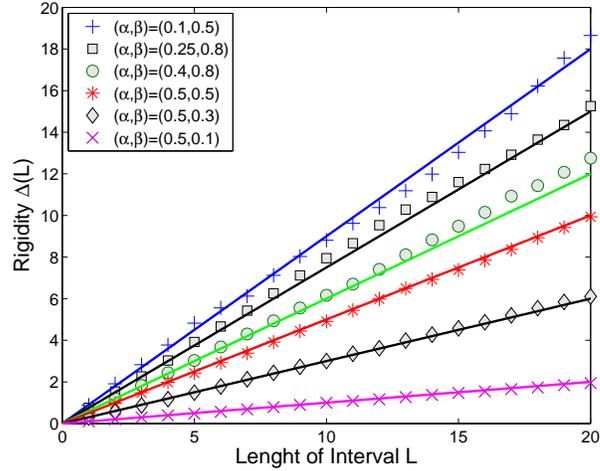,height=2.5in} \parbox{12cm}{\caption{Spectral Rigidity of TASEP-model. Signs represent the numerically-obtained spectral rigidity $\NV(L)$ for TASEP with parameters $\al,\be$ indicated in the legend. The curves display the corresponding rigidity calculated via the exactly-derived formula \eqref{NV-final-formula}.}}
\end{center}
\label{fig:TASEP_rigidity}
\end{figure}

\section{Summary and conclusions}

Using the MPA formalism we have derived the universal analytical formula for distance clearance among subsequent particles of TASEP with open boundary conditions. Such an exact formula has been substantially simplified for middle segment of sufficiently long TASEP-chains and consecutively compared to the empirically-obtained clearance distribution of vehicular streams. As uncovered, the interactions among the TASEP-particles are extremely short-ranged and hard-cored, which leads to the fact that the continuous (and re-scaled) alternative for associated clearance distribution \eqref{universal_headway_distribution} reads as
\BE \wp(r)=\Theta(r)\e^{-r}. \label{continuous_clearance} \EE
Thus, in spite of the fact that the TASEP-particles are repulsed such interactions are significantly different from those detected in realistic traffic flow. Anyway, the traffic clearance is of the form \eqref{p_beta}, which means that distribution \eqref{continuous_clearance} corresponds to the traffic clearance distribution only if the traffic density is small ($\nu \approx 0$). But on the other hand, if the spectral rigidities are compared one can detect a good agreement between the TASEP-rigidity \eqref{NV-final-formula} calculated analytically and the same quantity analyzed in traffic data (compare Figures 5 and 6).\\

To conclude, although the traffic model based on the TASEP-rules shows similar macroscopical effects as systems of moving vehicles it is definitely not suitable for modeling the microstructure of traffic (especially in congested traffic regimes).\\

\emph{Acknowledgements:} The authors would like to thank Bogolyubov Laboratory of Theoretical Physics, Joint Institute of Nuclear Research in Dubna (Russia) for incentives to the presented research. This work was supported by the Ministry of Education, Youth
and Sports of the Czech Republic within the project MSM 6840770039.

\section{Appendix}

\begin{lemma}
\label{lemma:com}
    Let $m,n,a$ be the arbitrary natural numbers. Let $ a \leq m.$ Then
        \BE
        \label{eq:lemma,com}
            \sum_{i=a}^m {n-i+m-1\choose n-1}={n+m-a\choose n}.
        \EE
\end{lemma}
\begin{proof}
    We will proceed using the mathematical induction. Let $m=a=1.$ Then for an arbitrary $n\in\mathbb{N}$
$$\sum_{i=1}^1{n-i \choose n-1}={n-1\choose n-1}={n \choose n}.$$
Now let  the relation~\eqref{eq:lemma,com} hold for all $a\leq m.$ Then for $a< m+1$
$$\sum_{i=a}^{m+1}{n-i+m\choose n-1}=\sum_{i=a}^{m}{n-i+m-1\choose n-1}+{n-a+m\choose n-1}={n+m-a\choose n}+{n-a+m\choose n-1}={n-a+m+1\choose n}.$$
The case $a=m+1$ can be proven analogically to the case $m=1,a=1.$
\end{proof}

\begin{lemma} If using the notation
\label{lemma:bnp}
    \BE
    \label{eq:notation}
        \B_{n,p}=
        \begin{cases}
            \frac{p(2n-p-1)!}{n!(n-p)!}& 0<p\leq n,\\
            0&\text{otherwise,}
        \end{cases}
    \EE
it can simply be verified that for $n>0$ and $p\geq0$
$$\B_{n,p}=\B_{n+1,p+1}-\B_{n+1,p+2}.$$
\end{lemma}

\begin{lemma} For arbitrary square matrices $D,E$ fulfilling $DE=D+E$ and for $n\geq0$
\label{lemma:dnde}
$$D^{n}(D+E)=D^{n+1}+D^{n-1}(D+E)=\sum_{q=1}^{n+1}D^q+E.$$
\end{lemma}

\begin{lemma}
\label{lemma:d+en}
    Assume that $D,~E$ to be arbitrary square matrices fulfilling the relation $DE=D+E.$ Let $n$ be the natural number. Then
    \BE
        (D+E)^n =\sum_{m=1}^n \B_{n,m}\sum_{q=0}^m E^qD^{m-q}.\label{Dokazal_Derrida}
    \EE
\end{lemma}
\begin{proof} Again, we will use the mathematical induction with respect to $n$. Let us assume that~\eqref{Dokazal_Derrida} holds. Then with help of  lemma~\ref{lemma:bnp} and lemma~\ref{lemma:dnde} we elementarily deduce that
    \begin{multline*}
        (D+E)^{n+1}=\sum_{m=1}^n \B_{n,m}\sum_{q=0}^m E^qD^{m-q}(D+E)       =\sum_{m=1}^{n}\B_{n+1,m+1}\sum_{q=0}^{m}\left(E^{q+1}+E^{q}\sum_{p=1}^{m-q+1}D^q\right)-\\        -\sum_{m=1}^{n}\B_{n+1,m+2}\sum_{q=0}^{m}\left(E^{q+1}+E^{q}\sum_{p=1}^{m-q+1}D^q\right)        =\B_{n+2,2}(D+E)^2+\sum_{m=3}^{n+1}\B_{n+1,m}\sum_{q=0}^m E^qD^{m-q}=\\
        =\sum_{m=1}^{n+1} \B_{n+1,m}\sum_{q=0}^m E^qD^{m-q}.
    \end{multline*}
\end{proof}

\begin{lemma}
\label{lemma:den}
    Assume that $D,~E$ to be arbitrary square matrices fulfilling the relation $DE=D+E.$ Let $n\in\mathbb{N}$. Then
    \BE
    \label{eq:den}
        DE^n=DE^{n-1}+E^n=D+\sum_{w=1}^{n}E^w.
    \EE
\end{lemma}

\begin{lemma}
\label{lemma:dnem}
    Assume that $D,~E$ to be arbitrary square matrices fulfilling the relation $DE=D+E.$ Let $m,n$ be the natural numbers. Then
    \BE
    \label{eq:dnem}
        D^n E^m =\sum_{i=1}^m {n-i+m-1\choose n-1}E^i+\sum_{j=1}^n {n-j+m-1\choose m-1}D^j.
    \EE
\end{lemma}
\begin{proof}
    We will use the mathematical induction. For $n=1$ and an arbitrary $m\in\mathbb{N}$ the expression~\eqref{eq:dnem} reduces to the expression~\eqref{eq:den}. Let us now assume that~\eqref{eq:dnem} holds true for some $n\in\mathbb{N}.$ By means of lemma~\ref{lemma:com} and lemma~\ref{lemma:den} we get
    \begin{multline*}
        D(D^nE^m)=\sum_{i=1}^m {n-i+m-1\choose n-1}\left[D+\sum_{w=1}^{i}E^w\right]+\sum_{j=1}^n {n-j+m-1\choose m-1}D^{j+1}=\\
        =\sum_{j=1}^{n+1} {n-j+m\choose m-1}D^j+\sum_{w=1}^{m}E^w\sum_{i=w}^m{n-i+m-1\choose n-1}=\sum_{i=1}^m {n-i+m\choose n}E^i+\sum_{j=1}^{n+1} {n-j+m\choose m-1}D^j.
    \end{multline*}
\end{proof}

\begin{figure}[htb]
\begin{center}
\epsfig{file=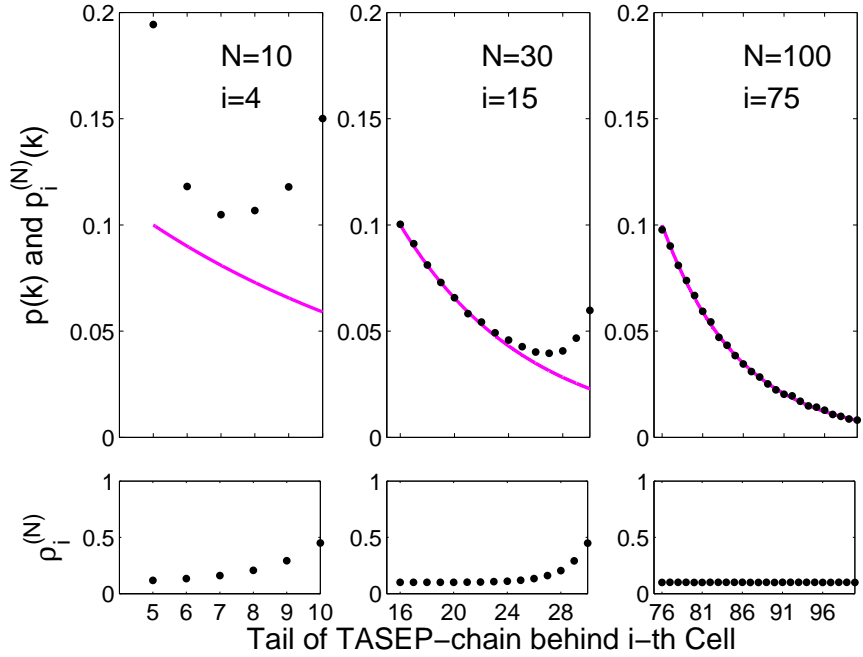,height=3.5in} \parbox{12cm}{\caption{Clearance Distribution in the Phase of Low Density I. We plot the inter-particle distance distribution calculated for TASEP with $\al=0.1$ and $\be=0.2.$ The bullets visualize the clearance distribution $\wp_i^{(N)}(k)$ enumerated via the exact formula \eqref{finalka-clearance}, while continuous curves visualize the large $N$ approximations $\wp(k)$ summarized by the relation \eqref{universal_headway_distribution}. Above that, we also plot the relevant densities, which synoptically illustrates the reasons for discrepancies between \eqref{finalka-clearance} and \eqref{universal_headway_distribution}.}}
\end{center}
\end{figure}

\begin{figure}[htb]
\begin{center}
\epsfig{file=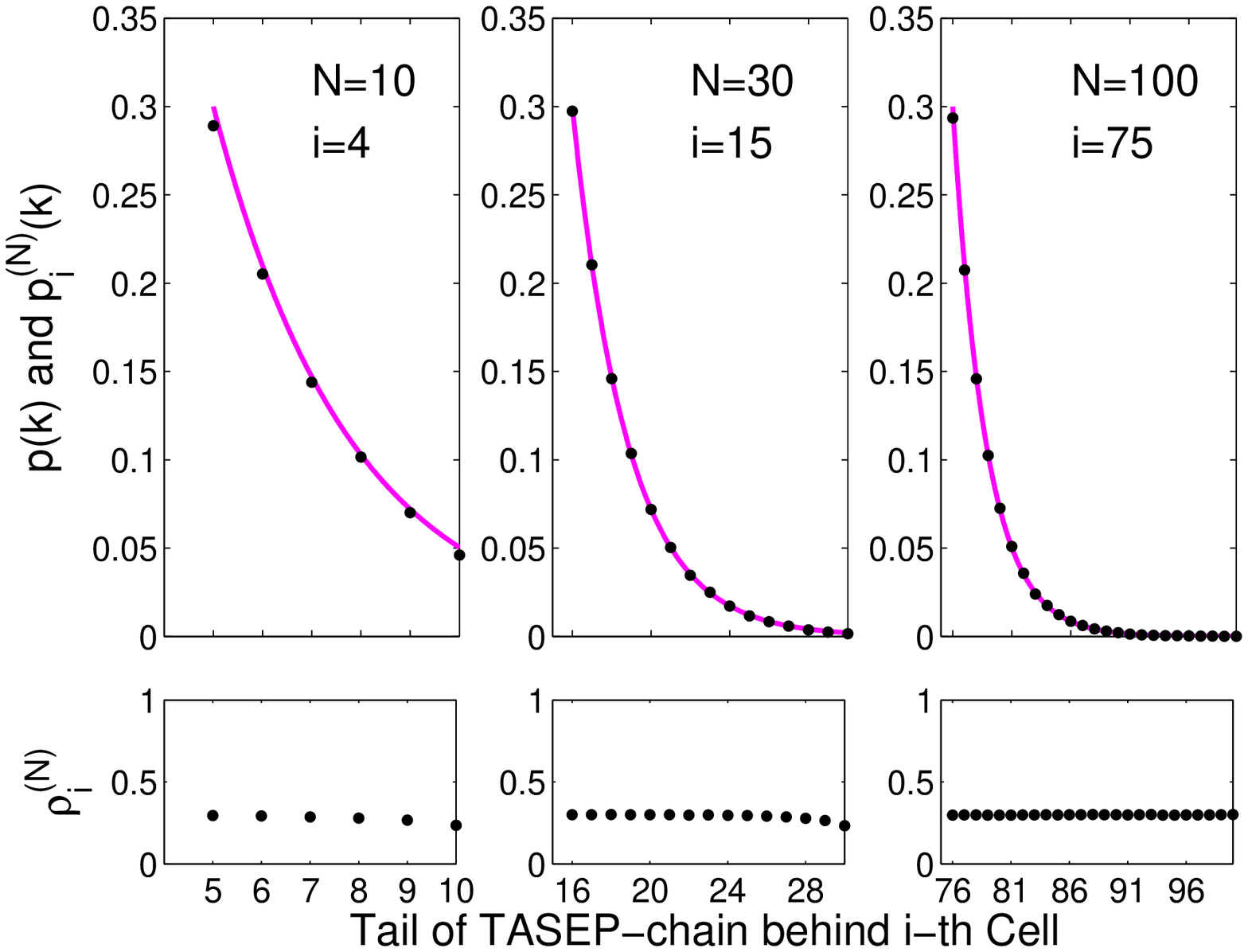,height=3.5in} \parbox{12cm}{\caption{Clearance Distribution in the Phase of Low Density II. We plot the inter-particle distance distribution calculated for TASEP with $\al=0.3$ and $\be=0.9.$ }}
\end{center}
\end{figure}

\begin{figure}[htb]
\begin{center}
\epsfig{file=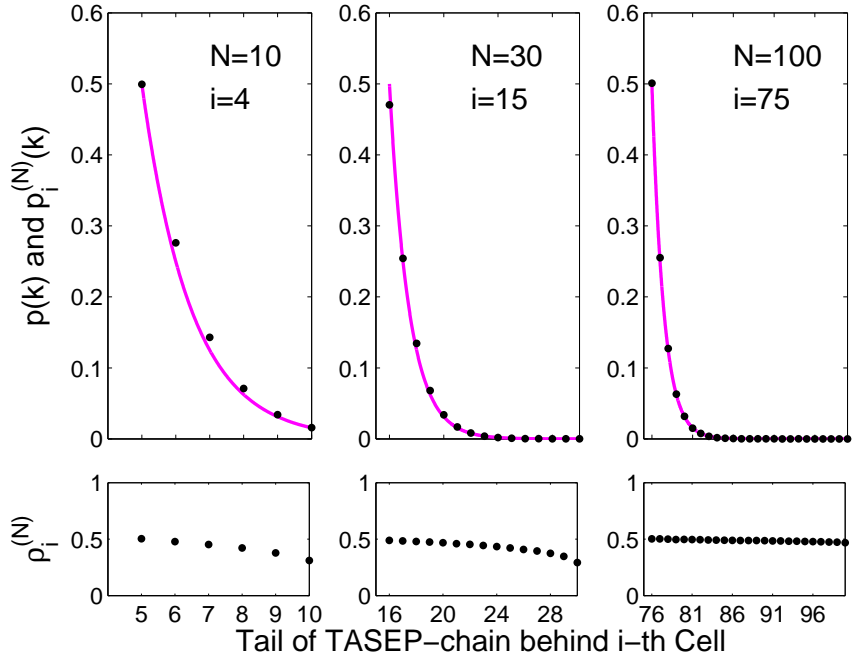,height=3.5in} \parbox{12cm}{\caption{Clearance Distribution in the Phase of Maximum Current. We plot the inter-particle distance distribution calculated for TASEP with $\al=0.8$ and $\be=0.9.$}}
\end{center}
\end{figure}

\begin{figure}[htb]
\begin{center}
\epsfig{file=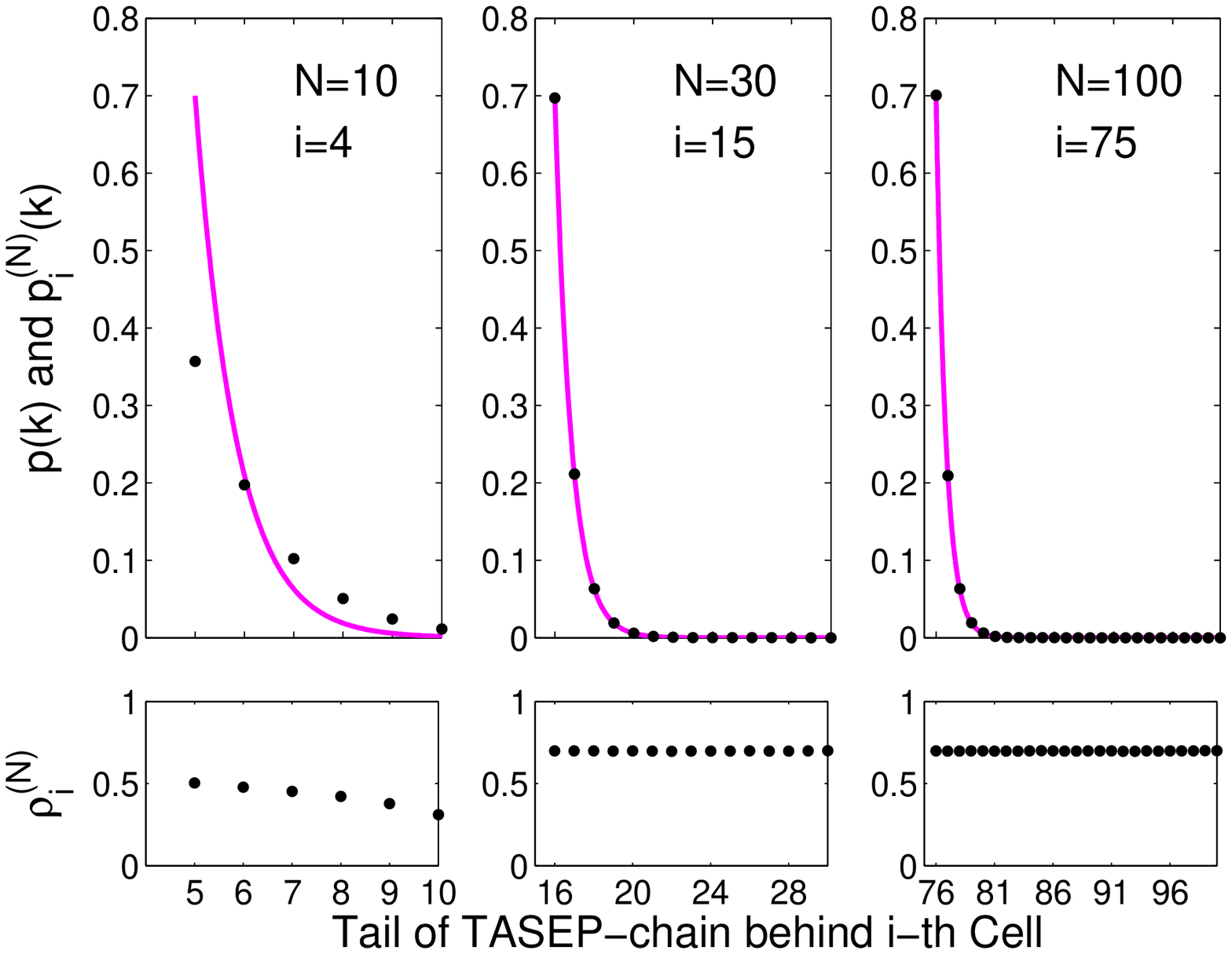,height=3.5in} \parbox{12cm}{\caption{Clearance Distribution in the Phase of High Density I. We plot the inter-particle distance distribution calculated for TASEP with $\al=0.9$ and $\be=0.3.$ }}
\end{center}
\end{figure}

\begin{figure}[htb]
\begin{center}
\epsfig{file=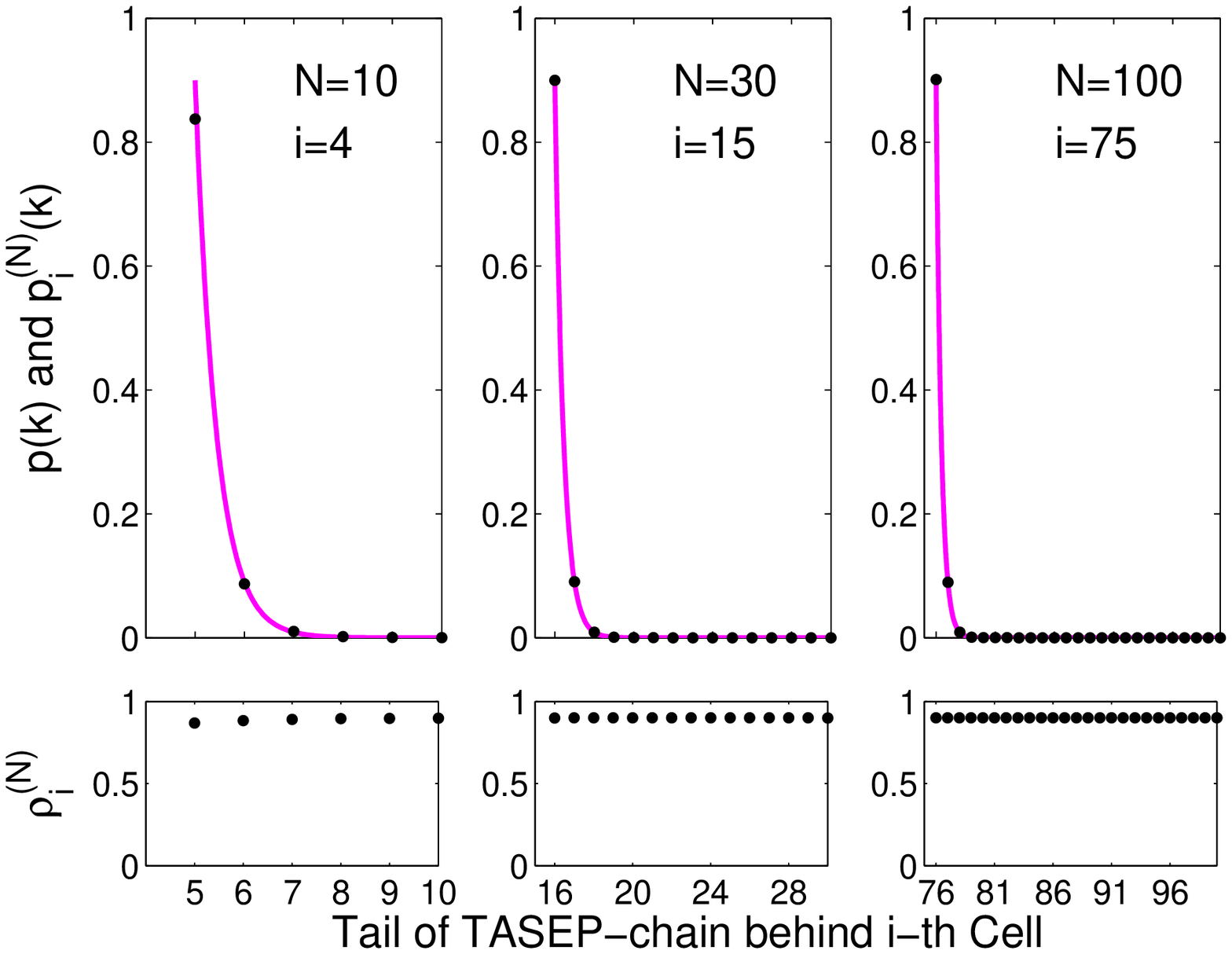,height=3.5in} \parbox{12cm}{\caption{Clearance Distribution in the Phase of High Density II. We plot the inter-particle distance distribution calculated for TASEP with $\al=0.2$ and $\be=0.1.$ }}
\end{center}
\end{figure}

\end{document}